 \documentclass[10pt, journal]{IEEEtran}

\usepackage{graphicx}
\usepackage{amssymb}
\usepackage{amsmath}%
\usepackage{amsfonts}%
\usepackage{amssymb}%
\usepackage{amsthm}
\usepackage{graphicx}
\usepackage{subfigure}
\usepackage{mathrsfs}
\usepackage[usenames,dvipsnames]{xcolor}
\usepackage{mdframed}
\usepackage{bm}

\newtheorem{thm}{Theorem}
\newtheorem{lem}{Lemma}

\newtheorem{cor}{Corollary}
\newtheorem{assumptions}{Assumption}

\theoremstyle{definition}

\newtheorem{defn}{Definition}

\newtheorem{remark}{Remark}

\newcommand{\etal}{\textit{et al.}~}

\newcommand{\PR}[1]{\mbox{Pr}\left\{ #1 \right\}}

\newcommand{\calX}{\mathcal{X}}

\newcommand{\calY}{\mathcal{Y}}
\newcommand{\calT}{\mathcal{T}}

\newcommand{\calJ}{\mathcal{J}}

\newcommand{\calS}{\mathcal{S}}
\newcommand{\calU}{\mathcal{U}}
\newcommand{\calI}{\mathcal{I}}
\newcommand{\calF}{\mathcal{F}}

\newcommand{\bz}{\mathbf{z}}

\newcommand{\bbz}{\mathbf{z}}

\renewcommand{\tilde}{\widetilde}

\newcommand{\Adv}{\mathsf{Adv}}

\newcommand{\Xh}{\hat{X}}

\DeclareMathOperator*{\argmax}{\arg\!\max}

\newcommand{\brk}[1]{\langle #1 \rangle}
\newcommand{\Reals}{\mathbb{R}}
\newcommand{\normQ}[1]{\| #1 \|_Q}

\newcommand{\normEuc}[1]{\| #1 \|_2}
\newcommand{\ox}{\bar{x}}
\newcommand{\ones}{\mathbf{1}}
\newcommand{\inertia}{\mathcal{I}}
\newcommand{\defined}{\triangleq}
\newcommand{\Tr}[1]{\mathrm{ tr}\left(#1 \right)}
\newcommand{\diag}[1]{\mathrm{diag}\left( #1 \right)}
\newcommand{\pxy}{p_{X,Y}}
\newcommand{\px}{p_X}
\newcommand{\py}{p_Y}
\newcommand{\pxp}{p_{X'}}

\newcommand{\pygx}{p_{Y|X}}
\newcommand{\Ppygx}[1]{\mathbf{p}_{Y|X={#1}}}
\newcommand{\pxhgx}{p_{\Xh|X}}
\newcommand{\qx}{q_X}

\newcommand{\ExpVal}[2]{\mathbb{E}_{#1}\left[ #2 \right]}

\newcommand{\Pxy}{P_{X,Y}}
\newcommand{\Pygx}{P_{Y|X}}
\newcommand{\Pxxp}{P_{X,X'}}
\newcommand{\Pxhgx}{P_{\hat{X}|X}}
\newcommand{\Px}{\mathbf{p}_X}
\newcommand{\Qx}{\mathbf{q}_X}

\newcommand{\Pxp}{\mathbf{p}_{X'}}
\newcommand{\Py}{\mathbf{p}_Y}
\newcommand{\At}{\tilde{A}}
\newcommand{\Bt}{\tilde{B}}
\newcommand{\Ut}{\tilde{U}}
\newcommand{\Vt}{\tilde{V}}
\newcommand{\lambdat}{\tilde{\lambda}}
\newcommand{\Sigmat}{\tilde{\Sigma}}
\newcommand{\by}{\mathbf{y}}

\newcommand{\blambda}{\pmb{\lambda}}
\newcommand{\bLambda}{\pmb{\Lambda}}
\newcommand{\Emat}{F}
\newcommand{\bu}{\mathbf{u}}
\newcommand{\bv}{\mathbf{v}}
\newcommand{\ba}{\mathbf{a}}

\newcommand{\tu}{\tilde{u}}
\newcommand{\tv}{\tilde{v}}
\newcommand{\olU}{\overline{U}}

\newcommand{\sto}{\mbox{\normalfont s.t.}}
\newcommand{\KFnorm}[2]{\| #1 \|_{#2}}

\newcommand{\bigO}{O}
\newcommand{\ttheta}{\tilde{\theta}}

\definecolor{light-gray}{gray}{.9}

\newmdtheoremenv[%
  backgroundcolor=light-gray, %
  linecolor=black,
  linewidth =1pt,%
  skipabove = 10pt,%
  skipbelow = 10pt
]{work}{Work in progress}

\newmdtheoremenv[%
  fontcolor=BrickRed,
  linecolor=black,
  linewidth =1pt,%
  skipabove = 10pt,%
  skipbelow = 10pt
]{todo}{Place holder}

\newmdtheoremenv[%
  fontcolor=BrickRed,
  linecolor=black,
  linewidth =1pt,%
  skipabove = 10pt,%
  skipbelow = 10pt
]{disc}{Disclaimer}

\IEEEoverridecommandlockouts

\title{Bounds on inference}
\author{Fl\'avio P. Calmon, Mayank Varia, Muriel M\'edard, \\ Mark M. Christiansen,
  Ken R. Duffy, Stefano Tessaro
\thanks{This work is sponsored by the Intelligence Advanced Research Projects
Activity under Air Force Contract FA8721-05-C-0002.  Opinions, interpretations,
conclusions and recommendations are those of the authors and are not necessarily
endorsed by the United States Government. M.C.
and K.D. are supported by Science Foundation Ireland Grant No. 
11/PI/1177.

\indent F.~P.~Calmon and M.~M\'edard are with the Research Laboratory
of Electronics at the Massachusetts Institute of Technology, Cambridge, MA (e-mail:
flavio@mit.edu; medard@mit.edu).
\newline 
\indent M.~Varia is with  MIT Lincoln Laboratory, Lexington, MA (e-mail:
mayank.varia@ll.mit.edu).
\newline 
\indent M.~M.~Christiansen and K.~R.~Duffy are with the Hamilton Institute
at the National University of Ireland, Maynooth, Ireland (e-mail:
mark.christiansen@nuim.ie; ken.duffy@nuim.ie).
\newline 
\indent S.~Tessaro is with the Computer Science and Artificial
Intelligence Laboratory at the Massachusetts Institute of Technology, Cambridge,
MA (e-mail: tessaro@csail.mit.edu).
}
}
\date{}                                           

\begin{document}
\maketitle

\begin{abstract}
Lower bounds for the average probability of error of estimating a hidden
variable $X$ given an observation of  a correlated random variable $Y$, and
Fano's inequality in particular, play a central role in information theory. In
this paper, we present a lower bound for the average estimation error based on
the marginal distribution of $X$ and the principal inertias of the joint
distribution matrix of $X$ and $Y$. Furthermore, we discuss an  information
measure based on the sum of the largest principal inertias,   called
$k$-correlation, which generalizes maximal correlation. We show that
$k$-correlation satisfies the Data Processing Inequality and is convex in the
conditional distribution of $Y$ given $X$. Finally, we investigate how to
answer a fundamental question in inference and privacy: given an observation
$Y$, can we estimate a function $f(X)$ of the hidden random variable $X$ with an
average error below a certain threshold? We provide a general method for
answering this question using an approach based on rate-distortion theory.
\end{abstract}

%
%
%

\section{Introduction}

Consider the standard problem in estimation theory: Given an observation of a
random variable $Y$, what can we learn about a correlated, hidden variable $X$?
For example, in security systems, $X$ can be the plaintext message, and  $Y$
 the ciphertext or any additional side information available to an
adversary.  Throughout the paper, we assume that $X$ and $Y$ are discrete random
variables with finite support.

If the joint distribution between $X$ and $Y$ is known, the  probability of
error of estimating $X$ given an observation of $Y$ can be calculated exactly.
However, in most practical settings, this joint distribution is unknown.
Nevertheless, it might be possible to estimate certain correlation measures of
$X$ and $Y$  reliably, such as  maximal correlation, $\chi^2$-statistic or
mutual information.

Given  estimates of such correlation measures, is it possible to determine a
lower bound for the average error probability of estimating $X$ from $Y$ over
all possible estimators? We answer this question in the affirmative. In the
context of security, this bound might characterize the best  estimation of the
plaintext that a (computationally unbounded) adversary can make given an
observation of the output of the system.

Furthermore, owing to the nature of the joint distribution, it may be infeasible
to estimate $X$ from $Y$ with small error probability. However, it is possible
that a non-trivial function   $f(X)$ exists that is of interest to a learner and
can be  estimated reliably.   If $f$ is the identity function, this  reduces to
the standard problem of estimating $X$ from $Y$.  Determining if such a function
exists is relevant to several applications in inference, privacy and security
\cite{goldwasser_probabilistic_1984}.


In this paper, we establish lower bounds for the  average estimation error of
$X$ and $f(X)$ given an observation of $Y$. These bounds depend only on certain
measures of information between $X$ and $Y$ and the marginal distribution of
$X$. The results hold for any estimator, and they shed light on the fundamental
limits of what can be inferred about a hidden variable from a noisy measurement.
The bounds derived here are similar in nature to Fano's inequality  
\cite{cover_elements_2006},    
 and can be characterized as the solution of a convex program which, in turn, is
closely related to the rate-distortion optimization problem. 

Our work has two main contributions. First, we analyze properties of a measure
of information (correlation) between $X$ and $Y$ based on the principal
inertias of the joint distribution of $X$ and $Y$. The estimation of principal
inertias is widely studied in the field of correspondence analysis, and
is used in practice to analyze categorical data.  The metric we propose, called
\textit{$k$-correlation}, is defined as the sum of the $k$ largest principal inertias, which, in
turn, are the singular values of a particular decomposition of the joint
distribution matrix of $X$ and $Y$. We show that $k$-correlation generalizes
both the maximal correlation and the $\chi^2$ measures of correlation.
We also prove that $k$-correlation satisfies two key properties for information
measures: (i) the Data Processing Inequality and (ii) convexity in the
conditional probabilities $\pygx$. Furthermore, we  derive a family of lower
bounds for the average  error probability of estimating $X$ given $Y$ based on
the principal  inertias between $X$ and $Y$ and the marginal distribution of
$X$.

The second  contribution is a general procedure for bounding the average
estimation error of  a deterministic function of $X$ given an observation of
$Y$. These bounds are non-trivial and help characterize the
fundamental limits of what can be learned about $X$ given an observation of $Y$.
For example, given  $I(X;Y)\leq \theta$, a positive integer $M$ and the marginal
distribution of $X$, this procedure allows us to compute a lower bound for the
average estimation error of  any surjective function  that maps the
support of $X$ onto $\{1,\dots,M\}$.

The rest of the paper is organized as follows. Section \ref{sec:backandnot}
presents an overview of the main results and discusses related work. Section
\ref{sec:measureMoments}  introduces the $k$-correlation metric of information,
and proves that it is convex in the transition probability $\pygx$ and satisfies
the Data Processing Inequality. Section \ref{sec:fanoinertia} presents a
Fano-like inequality based on the  principal inertias and the marginal
distribution $\px$. Section \ref{sec:functions} presents a general method for
deriving bounds for the average estimation error of deterministic surjective
functions of $X$ from an observation of $Y$.  Finally, concluding remarks are
presented in section \ref{sec:conclusion}.  


\section{Overview of main results and related work}
\label{sec:backandnot}

\subsection{Notation}
We assume throughout this paper that, for a given sample space $\Omega$,
$X:\Omega\rightarrow \calX$ is the hidden random variable and
$Y:\Omega\rightarrow \calY$ is the observed random variable, where   $
\calX=\{1,\dots,m\}$ and $ \calY=\{1,\dots,n\}$ are the  respective support sets.
We denote by $\Pxy$ and $\Pygx$ the $m\times n$ matrices with entries $
\left[\Pxy\right]_{i,j}\defined \pxy(i,j)$ and $\left[\Pygx\right]_{i,j}\defined
\pygx(j|i)$, respectively.  Furthermore, we denote by $\Px\in \Reals^m$,
 $\Py\in \Reals^n$ and $\Ppygx{j} \in \Reals^n$  the column vectors with entries
\begin{equation*}
\left[\Px\right]_{i}\defined p_{X}(i),~  \left[\Py\right]_{i}\defined p_{Y}(i)\mbox{~and~}
\left[\Ppygx{j}\right]_{i}\defined \pygx(i|j),
\end{equation*}
respectively. The diagonal matrices with entries $\px$ and $\py$ are represented
as $D_X = \diag{\Px}$ and $D_Y=\diag{\Py}$. For a discrete random variable $Z$,
we denote by $X\rightarrow Y \rightarrow Z$  the fact that $p_{X,Y,Z}(x,y,z)
=\px(x)\pygx(y|x)p_{Z|Y}(z|y)$ (i.e. $X,Y,Z$ form a Markov chain). 

Given an observation of $Y$, the  estimation problem considered here is to find
a function $h(Y)=\hat{X}$ that minimizes the average error probability $P_e$,
defined as 
\begin{equation}
    P_e \defined \Pr\left\{ \hat{X} \neq X \right\}.
    \label{eq:probError}
\end{equation}
Note that $X\rightarrow Y \rightarrow \hat{X}$.  $P_e$ is minimized when
$\hat{X}$ is the maximum-likelihood estimate of $X$.

The column vector with all entries  equal to 1 is represented by $\ones$.
The length of the vector will be clear from the context. 
For any given matrix
$A$, we denote by $\sigma_k(A)$ the $k$-th largest singular value of $A$. If $A$
is hermitian, we denote the $k$-th largest eigenvalue of $A$ by $\Lambda_k(A)$.
We denote by $\calS^m_{++}$  the set of positive definite matrices in
$\Reals^{m\times m}$. Furthermore,
\begin{equation}
  \calT_{m,n} \defined \left\{A\in \Reals^{m\times n}: \mbox{$A$ is row-stochastic, }
  [A]_{i,j}\geq 0  \right\}.
\end{equation}

For a given measure of information (correlation) $\calI(X;Y)$ between $X$ and
$Y$ (such as mutual information or maximal correlation), we denote
$\calI(X;Y)=\calI(\px,\Pygx)$  when we wish to highlight $\calI(X;Y)$ as a
function $\px$ and the transition matrix $\Pygx$.


\subsection{Overview of main results}

Assume that the joint distribution  $\pxy$ is unknown, but that the marginal
distribution $\px$ is given.   Furthermore, assume that  a certain measure of
information (correlation) $\calI(X;Y)$ between $X$ and $Y$ is bounded above by
$\theta$, i.e. $\calI(X;Y)\leq \theta$.  In practice, the value of $\theta$
and $\px$ could be determined, for example, from multiple i.i.d. samples  drawn
according to $\pxy$. The number of samples available might be insufficient to
characterize $\pxy$, but enough to estimate
$\theta$ and $\px$ reliably. Under these assumptions, what can be said about the
smallest $P_e$ possible? Our goal in this paper is to derive lower bounds of the
form $P_e\geq L_\calI(\px,\theta)$, creating a limit on how well $X$ can be
inferred from $Y$.

The characterization of $L_\calI(\px,\theta)$ for different measures of
information $\calI$ is particularly relevant for applications in privacy and
security, where $X$ is a variable that should remain hidden (e.g.
plaintext). A lower bound for $P_e$ can then be viewed as a security metric:
regardless of an adversary's computational resources, he will not be able to
guess $X$ with an average estimation   error smaller than $L_\calI(\px,\theta)$
given an observation of $Y$.
Therefore, by simply estimating $\theta$ and calculating $L_\calI(\px,\theta)$
we are able to evaluate the privacy threat incurred by an adversary
that has access to $Y$.  


If $\calI(X;Y)=I(X;Y)$, where $I(X;Y)$ is the mutual information between $X$ and
$Y$, then Fano's inequality \cite{cover_elements_2006} provides a lower bound
for $P_e$. However, in practice, several other statistics are  used 
in addition to mutual information in order to capture the information (correlation)
between $X$ and $Y$. In this work, we focus on one particular metric, namely the \textit{principal
inertia components} of $\pxy$, denoted by the
vector 
  $(\lambda_1,\dots,\lambda_d),$ 
where $d=\min\{m-1,n-1 \}$, and $\lambda_1\geq
\lambda_2\geq\dots\geq\lambda_d$. The exact definition of the principal inertias
is presented in Section \ref{sec:measureMoments}.

\subsubsection{Bounds based on principal inertia components}

 The principal inertias generalize other measures that are used in information
 theory. In particular, $\lambda_1 = \rho_m^2(X;Y)$, where $\rho_m(X;Y)$ is the
 \textit{maximal correlation} between $X$ and $Y$. Given
  \begin{align*}
  \calS \defined \left\{ (f(X),g(Y)):\right.&\ExpVal{}{f(X)}=\ExpVal{}{g(Y)}=0,\\
  &\left.\ExpVal{}{f^2(X)}=\ExpVal{}{g^2(Y)}=1 \right\},
\end{align*}
the maximal correlation $\rho_m(X;Y)$ is defined as \cite{renyi_measures_1959}
\begin{align*}
  \rho_m(X;Y) &= \max_{(f(X),g(Y))\in\calS}\ExpVal{}{f(X)g(Y)}.
\end{align*}
 In section \ref{sec:measureMoments} and
 appendix \ref{apx:interp}, we discuss how to compute the principal inertias
 and provide two alternative characterizations.  Compared to mutual information,
 the principal inertias  provide a finer-grained decomposition of the
 correlation between $X$ and $Y$. 

We propose a  metric of information called $k$-correlation, defined as
 $ \calJ_k(X;Y) \defined \sum_{i=1}^k \lambda_i$.
This metric satisfies two key
properties:

\begin{itemize}
  \item Convexity in $\pygx$ (Theorem \ref{thm:convex});
  \item Data Processing Inequality (Theorem \ref{lem:dataProc}). This is also satisfied by
      $\lambda_1,\dots,\lambda_d$ individually.
    \end{itemize}

By making use of the fact that the principal inertia components satisfy the Data
Processing Inequality, we are able to derive a family of bounds for $P_e$ in
terms of $\px$ and $\lambda_1,\dots,\lambda_d$, described in Theorem
\ref{thm:Bound}. This result sheds light on the relationship of $P_e$ with   the
principal inertia components. 

One immediate consequence of Theorem \ref{thm:Bound} is a useful scaling law for
$P_e$ in terms of the largest principal inertia (i.e. maximal
correlation). Let $X=1$ be the most likely outcome for $X$. Corollary
\ref{cor:coolBounds} proves that the advantage an adversary has of guessing
$X$, over the trivial solution of simply guessing the most
likely outcome of $X$ (i.e. $X=1$), satisfies
\begin{equation} 
  \Adv(X;Y) \defined \left|1-\px(1)-P_e\right|\leq O\left(\sqrt{\lambda_1}\right).
\label{eq:advGuess}
\end{equation}

\subsubsection{Bounding the estimation error of functions}

For most security applications, minimizing the probability that an adversary
guesses the hidden variable $X$ from an observation of $Y$ is insufficient.
Cryptographic definitions of security, and in particular semantic security
\cite{goldwasser_probabilistic_1984}, require that an adversary has negligible
advantage in guessing any function of the input given an observation of the
output. In light of this, we present bounds for the best
possible average error achievable for estimating 
functions of  $X$ given an observation of $Y$.

Still assuming that $\pxy$ is unknown, $\px$ is given and $\calI(X;Y)\leq
\theta$, we present in Theorem \ref{thm:boundPeM} a method for adapting bounds
of the form $P_e\geq L_\calI(\px,\theta)$ into bounds for the average estimation error of functions of $X$
given $Y$. This method depends on $\calI$ satisfying a few technical assumptions
(stated in section \ref{sec:functions}), foremost of which is the
existence of a lower bound $L_\calI(\px,\theta)$ that is
Schur-concave\footnote{A function $f:\Reals^n\rightarrow \Reals$ is said to be
  \textit{Schur-concave} if for all $x,y\in \Reals^n$ where $x$ is majorized by
$y$, then $f(x)\geq f(y)$.} in $\px$ for a fixed $\theta$. Theorem
\ref{thm:boundPeM} then states that, under these assumptions, for any deterministic,
surjective function $f:\calX\rightarrow \{1,\dots,M\}$, we can obtain a lower
bound for the average estimation error of $f$ by  computing
$L_\calI(p_U,\theta)$, where $U$ is a random variable that is a function $X$.

Note that Schur-concavity of $L_\calI(p_X,\theta)$ is crucial for this result.
In Theorem \ref{thm:schur}, we show that this condition is always satisfied when
$\calI(X;Y)$ is concave in $\px$ for a fixed $\pygx$, convex in $\pygx$ for a
fixed $\px$, and satisfies the Data Processing Inequality. This generalizes a
result by Ahlswede  \cite{ahlswede_extremal_1990} on the extremal properties of
rate-distortion functions.  Consequently, Fano's inequality can be  adapted in
order to bound the average estimation error of functions, as shown in
Corollary \ref{cor:PeMboundI}. By observing that a particular form of the bound
stated in Theorem \ref{thm:Bound} is Schur-concave, we also present a bound for
the error probability of estimating functions in terms of  the maximal
correlation, as shown in Corollary \ref{cor:PeMboundrho}.

\subsection{Background}

The joint distribution matrix $\Pxy$ can be viewed as a contingency table and
decomposed  using standard techniques from correspondence analysis
\cite{greenacre_correspondence_2007,greenacre_theory_1984}.  We note that this
decomposition was originally investigated by Hirschfield
\cite{hirschfeld_connection_1935}, Gebelein \cite{gebelein_statistische_1941} and
later by R\'enyi \cite{renyi_measures_1959}. For a quick overview of
correspondence analysis, we refer the reader to
\cite{greenacre_geometric_1987}.

The largest principal inertia of $\Pxy$ is equal to $\rho_m^2(X;Y)$, where
$\rho_m(X;Y)$ is the \textit{maximal correlation} between $X$ and $Y$. Maximal
correlation has been widely studied in the information theory and statistics
literature (e.g \cite{renyi_measures_1959}). Anantharam \etal present in
\cite{anantharam_maximal_2013} an overview of different characterizations of
maximal correlation, as well as its application in information theory. 

The Data Processing Inequality for the principal inertias was shown by Kang and
Ulukus in \cite[Theorem 2]{kang_new_2011} in a different setting than the one
considered here. Kang and Ulukus make use of the
decomposition of the joint distribution matrix  to derive outer bounds
for the rate-distortion region achievable in certain distributed source and channel
coding problems. 

Lower bounds on the average estimation error can be found using Fano-type
inequalities. Recently, Guntuboyina \etal
(\cite{guntuboyina_lower_2011,guntuboyina_sharp_2013}) presented a family of
sharp bounds for the minmax risk in estimation problems involving general
$f$-divergences. These bounds generalize Fano's inequality and, under certain
assumptions,  can be extended in order to lower bound $P_e$.

Most information-theoretic approaches for estimating or communicating functions
of a random variable are concerned with properties of specific functions  given
i.i.d. samples of the hidden variable $X$, such as in the functional compression
literature \cite{doshi_functional_2010,orlitsky_coding_2001}. These results
are rate-based and asymptotic, and do not immediately extend to the case where
the function $f(X)$ can be an arbitrary member of a class of 
functions, and only a single observation is available.

More recently,  Kumar and Courtade
\cite{kumar_which_2013} investigated boolean functions in an
information-theoretic context. In particular, they analyzed which is the most
informative (in terms of mutual information) 1-bit function (i.e. $M=2$) for the
case where $X$ is composed by $n$ i.i.d. Bernoulli(1/2) random variables, and
$Y$ is the result of passing $X$ through a discrete memoryless binary symmetric
channel. Even in this simple case, determining the most informative function is
non-trivial.

Bellare \etal \cite{bellare_semantic_2012}  considered  the standard wiretap
setting \cite{liang_information_2009}, and proved the equivalence between
semantic security and minimizing the maximum mutual information  over all
possible input message distributions. Since semantic security
\cite{goldwasser_probabilistic_1984} is achieved only when an adversary's
advantage of correctly computing a function of the hidden
variable given an observation of the output is negligibly small, the results in
\cite{bellare_semantic_2012}  are closely related to the ones presented here.


\section{A measure of information based on principal inertias}
\label{sec:measureMoments}

In this section we discuss how the joint probability matrix $\Pxy$ can be
decomposed into principal inertia components\footnote{The term \textit{principal
inertias} is borrowed from the correspondence analysis literature
\cite{greenacre_theory_1984}.}, and introduce the $k$-correlation measure. We
also prove that the $k$-correlation measure is convex in $\pygx$ and satisfies
the Data Processing Inequality. Several equivalent characterizations of the
principal inertias have appeared in the literature (e.g.
\cite{gebelein_statistische_1941}  and \cite{greenacre_geometric_1987}). We
discuss two of these characterizations  in appendix \ref{apx:interp}.


 Consider the singular value
 decomposition \cite{horn_matrix_2012} of the matrix $D_{X}^{-1/2} \Pxy
 D_Y^{-1/2}$, given by
\begin{equation}
    D_{X}^{-1/2} \Pxy D_Y^{-1/2}= U\Sigma V^T, 
    \label{eq:fullPxyDecomp}
\end{equation}
and define  $\At \defined D_X^{1/2}U$ and $\Bt \defined D_Y^{1/2}V$. Then 
\begin{equation}
 \Pxy = \At \Sigma \Bt^T,
 \label{eq:comactDecompPxy}
\end{equation} 
where $\At^T D_X^{-1} \At=\Bt^T D_Y^{-1} \Bt = I$. 

\begin{defn} 
  The square of the diagonal entries of $\Sigmat$ are called the \textit{principal
  inertias}, and are denoted by $\lambda_1,\dots,\lambda_d$, where
  $d=\min(m-1,n-1)$. Throughout this paper, we assume that principal inertias
  are ordered as $\lambda_1\geq \lambda_2\geq\dots\geq\lambda_d$.
  \label{defn:inertias}
\end{defn}

It can be shown that $\At$, $\Bt$ and $\Sigma$ have the form
\begin{align}
  \At = \left[ \Px~~A \right],~ \Bt = \left[ \Py~~B \right], \label{eq:defAB}\\
  \Sigma  =
  \diag{1,\sqrt{\lambda_1},\dots,\sqrt{\lambda_d}}, \nonumber
\end{align}
and, consequently, the joint distribution can be written as 
\begin{equation} 
  \pxy(x,y) = \px(x)\py(y)+\sum_{k=1}^d
  \sqrt{\lambda_k}b_{y,k}a_{x,k},
\end{equation}
where $a_{x,k}$ and $b_{y,k}$ are the entries of $A$ and $B$  in
\eqref{eq:defAB}, respectively,


Based on the decomposition of the joint distribution matrix, we  define below a
measure of information between $X$ and $Y$ based on the principal inertias.

\begin{defn}
  Let $\KFnorm{A}{k}$ denote the $k$-th Ky Fan norm\footnote{For $A\in
\Reals^{m\times n}$,  $\KFnorm{A}{k} =
\sum_{i=1}^k \sigma_i$, where $\sigma_1,\dots,\sigma_{\min\{m,n\}}$ are the
singular values of $A$.} \cite[Example 7.4.8]{horn_matrix_2012} of a matrix $A$. For $1\leq
  k \leq d$, we define the \textit{$k$-correlation} between $X$ and $Y$ as
  \begin{align}
    \calJ_k(X;Y)& \defined \KFnorm{D_X^{-1/2}\Pxy D_Y^{-1}\Pxy^T D_X^{-1/2} }{k}-1\\
               &= \sum_{i=1}^k \lambda_i.
  \end{align}
\end{defn}
Note that 
\begin{equation*}
    \calJ_1(X;Y)=\rho_m^2(X;Y),
\end{equation*}
where $\rho_m(X;Y)$ is the \textit{maximal correlation} of $(X,Y)$
\cite{anantharam_maximal_2013}, and
\begin{align*}
  \calJ_d(X;Y) = \ExpVal{X,Y}{\frac{\pxy(X,Y)}{\px(X)\py(Y)}}-1=\chi^2.
\end{align*}

We now show that $k$-correlation and, consequently, maximal correlation, is
convex in  $\pygx$ for a fixed $\px$ and satisfies the Data Processing
Inequality.

\subsubsection{Convexity in $\pygx$}
We use the next lemma to prove convexity of $\calJ_k(X;Y)$ in the transition
probability $\pxy$. 

\begin{lem} 
  For   $W\in \calS^m_{++}$ and $1\leq k \leq m$,
  the function $h_k: \Reals^{m\times n}\times \calS^n_{++}\to \Reals$ defined
  as
    \begin{equation}
      \label{eq:defhk}
      h_k(C,W)\defined \KFnorm{ CW^{-1}C^T }{k}
    \end{equation}
is convex.
\label{lem:convex}
\end{lem}
\begin{proof}
  Let $Q \defined  CW^{-1}C^T $.  Since $Q$ is positive semidefinite,
  $\KFnorm{Q}{k}$ is the sum of the $k$ largest eigenvalues of $Q$, and can be
  written as  \cite{fan_theorem_1949,overton_sum_1992}:
  \begin{equation}
    h_k(C,W) = \KFnorm{Q}{k} = \max_{Z^T Z= I_k} \Tr{Z^T Q Z}.
  \end{equation}
  Let $Z$ be fixed and $Z^T Z = I_k$, and denote the $i$-th column of $Z$ by
  $\bbz_i$. Note that $g(\ba,W)\defined \ba^T W^{-1} \ba$ is convex
  \cite[Example 3.4]{boyd_convex_2004} and, consequently, $g(C^T
   \bbz_i,W )$ is also convex in $C$ and $W$. Since the sum of convex
  functions is itself convex, then $\Tr{Z^TQ Z} = \sum_{i=1}^m g(C^T
   \bbz_i,W )  $ is also convex in $X$ and $Y$. The result follows by noting
  that the pointwise supremum over an infinite set of convex functions is also a
  convex function  \cite[Sec. 3.2.3]{boyd_convex_2004}.
\end{proof}

\begin{thm}
  \label{thm:convex}
    For a fixed $\px$, $\calJ_k(X;Y)$ is convex in $\pygx$.
\end{thm}
\begin{proof}
  Note that $\calJ_k(X;Y) = h_k(D_X \Pygx,D_Y)-1$, where $h_k$ is defined in
  equation \eqref{eq:defhk}.  For a fixed $\px$, $D_Y$ is a
  linear combination of $\pygx$. Therefore, since $h_k$ is convex (Lemma
  \ref{lem:convex}), and composition with an affine mapping preserves convexity,
  the result follows.
\end{proof}

\subsubsection{A data processing result}
In the next theorem, we prove that the principal inertias satisfy the Data
Processing Inequality.

\begin{thm}
  \label{lem:dataProc}
Assume that $X'\rightarrow X \rightarrow Y$, where $X'$
is a discrete random variable with finite support. Let $\lambda_1,\lambda_2,\dots,\lambda_d$ and
$\lambda_1',\lambda_2',\dots,\lambda_d'$
denote the principal inertias of $\Pxy$ and $P_{X',Y}$, respectively. Then
$\lambda_1\geq \lambda_1',\lambda_2\geq \lambda_2',\dots,\lambda_d \geq
\lambda_d'$.
\end{thm}
\begin{remark}
  This data processing result was also proved by Kang and Ulukus in
  \cite[Theorem 2]{kang_new_2011}, even though they do not make the explicit
  connection with maximal correlation and principal inertias. A weaker form of
  Theorem \ref{lem:dataProc} can be derived using a clustering result presented
  in \cite[Sec. 7.5.4]{greenacre_theory_1984} and originally due to Deniau \etal
  \cite{deniau_effet}. We use a different proof technique from the one in
  \cite[Sec. 7.5.4]{greenacre_theory_1984} and  \cite[Theorem 2]{kang_new_2011}
  to show result stated in the theorem, and present the  proof here for
  completeness. Finally, a related data processing result was stated in
  \cite[Eq. (31)]{yury_2013}.   
\end{remark}
\begin{proof}
  Assume without loss of generality that $\calX'=\{1,\dots,m'\}$  is the support
set of $X'$. Then $P_{X',Y}=F\Pxy $, where $F$ is a
$m'\times m$ column stochastic matrix. Note that $F$ represents the conditional
distribution of the mapping $X'\rightarrow X$, where the $(i,j)$-th
entry of $F$ is $p_{X'|X}(i|j)$.

Consider the decomposition of $P_{X',Y} = F\Pxy$:
\begin{align}
  S'&=  D_{X'}^{-1/2}\left( P_{X'Y}- \Pxp\Py^T \right)D_Y^{-1/2} \nonumber \\
& = D_{X'}^{-1/2}F \left( P_{XY}- \Px\Py^T \right)D_Y^{-1/2} \nonumber \\
  & =  D_{X'}^{-1/2}F D_X^{1/2} S \nonumber,
\end{align}
where $S$ is given by
\begin{align}
  S \defined D_X^{-1/2}\left( \Pxy- \Px\Py^T \right)D_Y^{-1/2}. \label{eq:SandP}
\end{align}
Note that the singular values of $S'$
are the principal inertias $\lambda_1',\dots,\lambda_d'$.

Let $E=D_{X'}^{-1/2}F D_X^{1/2} $, where
that the size of $E$ is $m'\times m$. Since $[FD_X]_{i,j}=p_{X',X}(i,j)$, then
the $(i,j)$-th entry of $E$ is
\begin{align*}
  [E]_{i,j} &=  \frac{p_{X',X}(i,j)}{\sqrt{p_{X'}(i)p_X(j)}}~.
\end{align*}
Observe that $E$ has the same form as \eqref{eq:fullPxyDecomp}, and, therefore,
$\KFnorm{E}{1}=1$. Let $H=  S^TS - S'^TS'$. Then for $\by\in \Reals^n$ and $S\by =
\bz$:
\begin{align*}
    \by^TH\by   &= \by^TS^TS\by - \by^TS'^TS'\by\\
    &= \normEuc{\bz} - \normEuc{E\bz}\\
    &\geq \normEuc{\bz}-\KFnorm{E}{1} \normEuc{\bz}\\
    &= 0.     
\end{align*}
Consequently, $H$ is positive semidefinite. Since $H$ is symmetric, it follows
from Weyl's theorem \cite[Theorem 4.3.1]{horn_matrix_2012} that for
$k=1,\dots,n$,
\begin{align*}
\Lambda_k(S'^TS') &\leq \Lambda_k(S'^TS' + H) \\
                  &= \Lambda_k(S^TS)\\
                  &= \lambda_k.
\end{align*}
Since $\Lambda_k(S'^TS')= \lambda'_k$, the result follows.

%
\end{proof}

The next corollary is a direct consequence of the previous theorem. 
\begin{cor}
    For $X' \rightarrow X \rightarrow Y$ forming a Markov chain, $J_k(X';Y)\leq J_k(X;Y)$.
\end{cor}


\section{A lower bound for the estimation error probability in terms of the principal
  inertias}
  \label{sec:fanoinertia}
Throughout the rest of the paper, we assume without loss of generality that
$\px$ is sorted in decreasing order, i.e. $\px(1)\geq\px(2)\geq \dots\geq
\px(m)$.

\begin{defn}
  Let $\bLambda(\Pxy)$ denote the vector of principal inertias of a joint
  distribution matrix $\Pxy$ sorted in decreasing order, i.e.
  $\bLambda(\Pxy)=(\lambdat_1,\dots,\lambdat_d)$.  We denote $\bLambda(\Pxy)\leq
  \blambda $ if $\lambdat_1\leq \lambda_1,\dots,\lambdat_d\leq \lambda_d$ and
  \begin{equation}
    \mathcal{R}(\Qx,\blambda)\triangleq\left\{\Pxy\big| \Px=\Qx\mbox{ and } \bLambda(\Pxy)\leq
  \blambda   \right\}.
\end{equation}
  \end{defn}

In the next theorem we present a Fano-like bound for the estimation error
probability of $X$ that depends on the marginal distribution $\px$ and on the
principal inertias. 
\begin{thm}
  \label{thm:Bound}
 For $\blambda = (\lambda_1,\dots,\lambda_d)$, define 
  \begin{align}
    k^* \defined \max \left\{ k\in \{1,\dots,m\} ~\big|~\px(k)-\Px^T\Px\geq
     0\right\}~, \label{eq:kstar}
   \end{align}
   \begin{align*}  
   f^*_0(\Px,\blambda)\defined&  \sum_{i=1}^{k^*}\lambda_i \px(i)+
    \sum_{i=k^* + 1}^{m}\lambda_{i-1} \px(i) \nonumber
    \\& - \lambda_{k^*} \Px^T\Px~, 
  \end{align*}     
\begin{align*}
    g_0(\beta,\Px,\blambda)&\defined f^*_0(\Px,\blambda)+\sum_{i=1}^m\left(\left[\px(i)-\beta\right]^+\right)^2,\\
    U_0(\beta,\Px,\blambda)&\defined \beta+
    \sqrt{ g_0(\beta,\Px,\blambda)},\\
    U_1(\Px,\blambda) &\defined \min_{0\leq\beta \leq \px(2)} U_0(\beta,\Px,\blambda).
  \end{align*}
Then
for any $\Pxy\in \mathcal{R}(\Px,\blambda)$, 
\begin{equation}
   P_e \geq 1-  U_1(\Px,\blambda).
   \label{eq:mainBound}
\end{equation}
\end{thm}
\begin{proof}
The proof of the theorem is presented in the appendix. 
\end{proof}

\begin{remark}
  If  $\lambda_i=1 $ for all $1\leq i \leq d$,   \eqref{eq:mainBound}
  reduces to $P_e \geq 0$. Furthermore, if  $\lambda_i=0$ for all $1\leq i \leq d$,  \eqref{eq:mainBound}
  simplifies to $P_e \geq 1- \px(1)$.
\end{remark}

We now present a few direct but powerful corollaries of the result in Theorem
\ref{thm:Bound}. We note that a  bound similar to \eqref{eq:gunt} below has
appeared in the context of bounding the minmax decision risk in
\cite[(3.4)]{guntuboyina_minimax_2011}. However, the proof technique used in
\cite{guntuboyina_minimax_2011} does not lead to the general bound presented in
Theorem \ref{thm:Bound}.

\begin{cor}
  If $X$ is uniformly distributed in $\{1,\dots,m\}$, then 
  \begin{equation}
    P_e \geq 1-\frac{1}{m}-\frac{\sqrt{(m-1)\chi^2}}{m}~.
    \label{eq:gunt}
  \end{equation}
  Furthermore, if only the maximal correlation $\rho_m(X;Y)=\sqrt{\lambda_1}$ is
  given, then
    \begin{align*}
      P_e &\geq 1- \frac{1}{m} -\sqrt{\lambda_1}\left(1-\frac{1}{m}\right)\\
                &=1-\frac{1}{m}-\rho_m(X;Y)\left(1-\frac{1}{m}\right).
    \end{align*}
\end{cor}

\begin{cor}
  \label{cor:coolBounds}
    For any pair of variables $(X,Y)$ with marginal distribution in $X$ equal to
    $\px$ and maximal correlation (largest
    principal inertia)   $\rho_m^2(X;Y)= \lambda_1$, we have for all
    $\beta\geq 0$
    \begin{equation}
      P_e \geq 1-  \beta -
      \sqrt{\lambda_1\left(1-\sum_{i=1}^m\px^2(i)\right)+\sum_{i=1}^m\left(\left[\px(i)-\beta\right]^+\right)^2}~.
    \label{eq:coolBound1}
    \end{equation}
In particular, setting $\beta = \px(2)$,
     \begin{align}
      P_e &\geq  1-  \px(2) -
      \sqrt{\lambda_1\left(1-\sum_{i=1}^m\px^2(i)\right)+\left(\px(1)-\px(2)\right)^2}
      \nonumber \\
      &\geq 1-  \px(1)-
      \rho_m(X;Y)\sqrt{\left(1-\sum_{i=1}^m\px^2(i)\right)}~.
      \label{eq:coolBound2}
    \end{align}
\end{cor}

\begin{remark} 
  The bounds \eqref{eq:coolBound1} and  \eqref{eq:coolBound2}   are
  particularly insightful in showing how the error probability scales with the
  input distribution and the maximal correlation. For a given $\pxy$, recall
  that  $\Adv(X;Y)$, defined in \eqref{eq:advGuess}, is the advantage of
  correctly estimating $X$ from an observation of $Y$ over a random guess of $X$
  when $Y$ is unknown.  Then, from equation \eqref{eq:coolBound2}
  \begin{align*}
        \Adv(X;Y) &\leq \rho_m(X;Y)\sqrt{\left(1-\sum_{i=1}^m\px^2(i)\right)} \\
        &= \bigO(\rho_m(X;Y)).
  \end{align*}
  Therefore, the advantage of estimating $X$ from $Y$ decreases at least linearly with
  the maximal correlation between $X$ and $Y$.
\end{remark}

\section{Lower bounds on estimating functions}
\label{sec:functions}

For any function $f:\calX\rightarrow \calU$, we denote by $\hat{f}$ the maximum-likelihood
estimator of $f(X)$ given an observation of $Y$. For a given integer $1\leq M\leq |\calX|$, we define
\begin{equation*}
  \calF_{M} \defined \left\{ f:\calX\rightarrow \calU~ \big|~f\mbox{ is
  surjective and }|\calU|=M
\right\}
\end{equation*}
and
\begin{equation}
  P_{e,M} \defined \min_{f\in \calF_M} \Pr\{f(X)\neq \hat{f}\}.
  \label{eq:PeM}
\end{equation}
$P_{e,|\calX|}$ is simply the error probability of estimating $X$ from
$Y$, i.e.  $P_{e,|\calX|}=P_e$. 

Throughout this section we adopt the following additional assumption.
\begin{assumptions}
  \label{assump:func}
An upper bound $\theta$ for a given measure of information $\calI(X;Y)$ between
$X$ and $Y$ is given, i.e.  $\calI(X;Y)\leq \theta$. Furthermore, $\calI(X;Y)$
satisfies the Data Processing Inequality, is convex in $\pygx$ for a given
$\px$, and is invariant to row and column permutations of the joint
distribution matrix $\pxy$. Finally, we also assume that the marginal
distribution of $X$, given by $\px$, is  known.
\end{assumptions}

Under this assumption, what can be said about $P_{e,M}$?  In the next sections
we present  a general procedure to derive non-trivial lower bounds for $P_{e,M}$.
\subsection{Extremal properties of the error-rate function}
Before investigating how to bound $P_{e,M}$, we first study how to bound  $P_e$
in a more general setting than the one in section \ref{sec:fanoinertia}. Note
that  
\begin{align*}
  P_e \geq &  \min_{\Pygx,E}~ 1-\Tr{ D_X\Pygx E}\\
  &~\sto~ \calI(\px,\Pygx)\leq \theta,~ \Pygx \in \calT_{m,n},~E \in
  \calT_{n,m}.  \nonumber 
\end{align*}
Here $E$ denotes the mapping from $Y$ to $\hat{X}$. By fixing $\Pygx$ and taking
the dual in $E$ of the previous convex program, we can verify that $E$ will
always be a row-stochastic matrix with entries equal to 0 or 1. Since
$\calI(X;Y)$ satisfies the Data Processing Inequality, $P_e \geq
e_\calI(\px,\theta)$, where $e_\calI(\px,\theta)$ is defined below.  

\begin{defn}
The \textit{error-rate function} $e_\calI(\px,\theta)$ is the
solution of the following convex program:
\begin{align}
  e_{\calI}(\px,\theta) \defined&  \min_{\Pxhgx}~ 1-\Tr{ D_X\Pxhgx} \label{eq:minPe} \\
  &~\sto~ \calI(\px,\Pxhgx)\leq \theta,~\Pxhgx \in \calT_{m,m}~.\nonumber
  \end{align}
\end{defn}

Due to convexity of $\calI(\px,\Pxhgx)$ in $\Pxhgx$, it follows directly that
$e_{\calI}(\px,\theta) $ is convex in $\theta$ for a fixed $\px$. Furthermore,
the cost function \eqref{eq:minPe} is equal to the average Hamming distortion
$\ExpVal{X,\hat{X}}{d_H(X,\hat{X})}$ between $X$ and $\hat{X}$.
Therefore, $e_\calI(\px,\theta)$ has a dual relationship\footnote{The authors thank
Prof. Yury Polyansky (MIT) for pointing out the dual relationship.} with the
rate-distortion problem 
\begin{align*}
  R_\calI(\px,\Delta)\defined& \min_{\Pxhgx} ~  \calI(\px,\Pxhgx)\\
  & \sto ~  \ExpVal{X,\hat{X}}{d_H(X,\hat{X})} \leq \Delta,  ~ \Pxhgx \in \calT_{m,m}. \nonumber
\end{align*}

We will now prove that, for a fixed $\theta$ (respectively, fixed $\Delta$),
$e_\calI(\px,\theta)$ (resp. $R_\calI(\px,\Delta)$) is \textit{Schur-concave}
in
$\px$ if $\calI(\px,\Pygx)$ is concave in $\px$ for a fixed $\Pygx$. Ahlswede
\cite[Theorem 2]{ahlswede_extremal_1990} proved this result for the particular
case where $\calI(X;Y)=I(X;Y)$  by investigating the properties of the explicit
characterization of the rate-distortion function under Hamming distortion. The
proof presented here is considerably simpler and more general, and is based on a
proof technique used by Ahlswede in \cite[Theorem 1]{ahlswede_extremal_1990}.

\begin{thm}
  \label{thm:schur}
If $\calI(\px,\Pygx)$ is concave in $\px$ for a fixed $\Pygx$, then
$e_\calI(\px,\theta)$ and $R_\calI(\px,\Delta)$ are Schur-concave in $\px$ for a
fixed $\theta$ and $\Delta$, respectively.
\end{thm}

\begin{proof}
Consider two probability distributions $\px$ and $\qx$ defined over
$\calX=\{1,\dots,m\}$. As usual, let $\px(1)\geq\px(2)\geq \dots\geq \px(m)$ and
$\qx(1)\geq \qx(2)\geq \dots\geq \qx(m)$.  Furthermore, assume that $\px$
majorizes $\qx$, i.e. $\sum_{i=1}^k \qx(i)\leq \sum_{i=1}^k \px(i)$ for $1\leq
k \leq m$. Therefore $q_X$ is a convex combination of permutations of $\px$
\cite{marshall_inequalities:_2011}, and can be written as $q_Z = \sum_{i=1}^l
a_i \pi_i \px$ for some $l\geq 1$, where $a_i\leq 0$, $\sum a_i = 1$ and $\pi_i$
are permutation operators, i.e. $\pi_i \px = p_{\pi_iX}$. Hence, for a fixed
$A\in \calT_{m,n}$:
\begin{align*}
\calI(\qx,A)&=\calI\left(  \sum_{i=1}^l a_i \pi_i \px,A\right)\\
             &\leq \sum_{i=1}^l a_i \calI(\pi_i \px,A),\\
             & =  \sum_{i=1}^l a_i \calI( \px,\pi_i A\pi_i),
\end{align*}
where the inequality follows from the concavity assumption and from $\calI(X;Y)$
being invariant to row and column permutations of the joint distribution matrix
$\pxy$. Consequently, from equation \eqref{eq:minPe}, 
\begin{align*}
  e_\calI(\qx,\theta)&= \inf_{A\in \calT_{m,m}}\left\{ 1- \sum_{i=1}^l a_i\Tr{D_X \pi_i
  A\pi_i}) :\right.\\
  & \hspace{0.7in} \left. \sum_{i=1}^l a_i \calI( \px,\pi_i A\pi_i)\leq \theta \right\} \\
  &\geq  \inf_{A_1,\dots,A_l\in \calT_{m,m}}\left\{ \sum_{i=1}^l a_i(1- \Tr{D_X
  A_i)} :\right.\\
  & \hspace{0.9in} \left. \sum_{i=1}^l a_i \calI( \px,A_i)\leq \theta \right\}\\
  & = \inf_{\theta_1,\dots,\theta_l\geq 0} \left\{ \sum_{i=1}^l a_i e_\calI
(\px,\theta_i): \sum_{i=1}^l a_i\theta_i = \theta \right\}\\
  &\geq \inf_{\theta_1,\dots,\theta_l\geq 0} \left\{e_\calI
\left(\px, \sum a_i \theta_i \right): \sum_{i=1}^l a_i\theta_i = \theta
\right\}\\
    &= e_\calI\left(\px, \theta \right),
\end{align*}
where the last inequality follows from the convexity of the error-rate function.
Since this holds for any $\qx$ that is majorized by $\px$, $
e_\calI(\px,\theta)$ is Schur-concave. Schur-concavity of  $R_\calI(\px,\Delta)$ follows directly
from its dual relationship with $e_\calI(\px,\theta)$.
\end{proof}

For $\calI = \calJ_k$, the convex program \eqref{eq:minPe} might be difficult to
compute due to the constraint on the sum of the singular values. The next
theorem presents a convex program that evaluates a lower bound for
$e_{\calJ_k}(\px,\theta)$ and can be solved using standard methods.

\begin{thm}
\begin{align}
  e_{\calJ_k}(\px,\theta) \geq &  \min_{\Pxhgx}~ 1-\Tr{ D_X\Pxhgx}\nonumber \\
  &~\sto~ \sum_{i=1}^k\sum_{j=1}^m \frac{\px(i)\pxhgx^2(j|i)}{y_j}  \leq
  \theta+1, \label{eq:constconvex} \\
  &~ ~~ ~ \Pxhgx \in \calT_{m,m},~ \nonumber \\
  & ~~\sum_{j=1}^m\px(i)\pxhgx(j|i)=y_j,~1\leq j \leq m. \nonumber
\end{align}
\end{thm}

\begin{proof}
Let $F\defined  D_{X}^{-1/2} P_{XY} D_Y^{-1/2}$. Then 
\begin{equation*}
\calJ_k(X;Y)= \KFnorm{FF^T}{k}-1.  
\end{equation*}
Let 
\begin{equation*}
  c_i\defined \sum_{j=1}^m \frac{\px(i)\pxhgx^2(j|i)}{y_j}
  \end{equation*}
be the $i$-th
diagonal entry of $FF^T$. By using the fact that the eigenvalues
 majorize the diagonal entries of a Hermitian matrix (\cite[Theorem
 4.3.45]{horn_matrix_2012}), we find
\begin{equation*}
  \sum_{i=1}^k c_i \leq \KFnorm{FF^T}{k},
\end{equation*}
and the result follows. Note that convexity  of the constraint
\eqref{eq:constconvex} follows from the fact that the perspective of a convex
function is convex \cite[Sec. 2.3.3]{boyd_convex_2004}.
\end{proof}

\subsection{Bounds for $P_{e,M}$}

Still adopting assumption \ref{assump:func}, a lower bound for $P_{e,M}$ can be
derived as long as  $e_\calI(\px,\theta)$ or a lower bound for $e_\calI(\px,\theta)$ is
Schur-concave in $p_X$. 

\begin{thm}
  \label{thm:boundPeM}
  For a given M, $1\leq M\leq m$, and $\px$, let $U=g(X)$, where $g_M:\{1,\dots,m\}\rightarrow
  \{1,\dots,M\}$ is defined as
  \begin{equation*}
    g_M(x) \defined 
        \begin{cases}
            1& 1\leq x \leq m-M+1\\
            x-m+M & m-M+2\leq x \leq m~.
        \end{cases}
  \end{equation*}
  Let $p_U$ be the marginal distribution\footnote{The pmf of $U$ is
    $p_U(1)=\sum_{i=1}^{m-M+1} \px(i)$ and $p_U(k)=\px(m-M+k)$ for
  $k=2,\dots,M$.} of $U$. Assume that, for a given measure of information
  $\calI(X;Y)$, there exists a function $L_{\calI}(\cdot,\cdot)$ such that for all
  distributions $q_X$ and any $\theta$,
  $e_\calI(\qx,\theta)\geq
  L_\calI(\qx,\theta)$.  Under assumption
  \ref{assump:func}, if $L_\calI(\px,\theta)$ is Schur-concave in $\px$, then
  \begin{equation}
    P_{e,M}\geq L_\calI(p_U,\theta)~.
  \end{equation} 

\end{thm}
\begin{proof}
  The result follows from the following chain of inequalities:
  \begin{align*}
    P_{e,M} & \stackrel{(a)}{\geq}   \min_{f\in \calF_M,\ttheta}
    \left\{e_\calI\left(p_{f(X)},\ttheta \right): 
    \ttheta\leq \theta \right\}\\
    & \geq \min_{f\in \calF_M}
    \left\{e_\calI\left(p_{f(X)},\theta\right)\right\}\\
    &  \stackrel{(b)}{\geq}  \min_{f\in \calF_M}
    \left\{L_\calI\left(p_{f(X)},\theta\right)\right\}\\
    & \stackrel{(c)}{\geq}  L_\calI(p_U,\theta),
  \end{align*}
  where (a) follows from the Data Processing Inequality, (b) follows from
  $e_\calI(\qx,\theta)\geq L_\calI(\qx,\theta)$ for all $\qx$, and $\theta$ and (c)
  follows from the Schur-concavity of the lower bound and by observing that  $p_U$
  majorizes $p_{f(X)}$ for every $f\in \calF_M$. 
\end{proof}

The following two corollaries illustrate how Theorem \ref{thm:boundPeM} can be
used for  different measures of information, namely mutual information and
maximal correlation.
\begin{cor}
  \label{cor:PeMboundI}
    Let $I(X;Y)\leq \theta$. Then
    \begin{equation*}
        P_{e,M}\geq d^* 
    \end{equation*}
    where $d^*$ is the solution of
    \begin{equation*}
      h_{b}(d^*)+d^*\log(m-1)=\min\{H(U)-\theta,0 \},
    \end{equation*}
    and $h_b(\cdot)$ is the binary entropy function.
\end{cor}
\begin{proof}
$R_I(\px,\delta)$ is  the well known rate-distortion function under Hamming
distortion, which satisfies (\cite[(9.5.8)]{gallager_information_1968})
$R_I(\px,\delta)\geq H(X)-h_{b}(d^*)-d^*\log(m-1)$. The result follows from
Theorem \ref{thm:schur}, since mutual information is concave in $\px$.
\end{proof}

\begin{cor}
    \label{cor:PeMboundrho}
    Let $\calJ_1(X;Y)=\rho_m(X;Y)\leq \theta$. Then  
    \begin{equation*} 
      P_{e,M} \geq 1-p_U(1)- \theta\sqrt{\left(1-\sum_{i=1}^M p_U^2(i)\right)}~.
   \end{equation*}
\end{cor}
\begin{proof}
  The proof follows directly from Theorems \ref{thm:convex}, \ref{lem:dataProc}
  and Corollary \ref{cor:coolBounds}, by noting that \eqref{eq:coolBound2}
  is Schur-concave in $\px$.
\end{proof}

\section{Concluding remarks}
\label{sec:conclusion}
We illustrated in this paper how the principal inertia-decomposition of the
joint distribution matrix can be applied to derive useful bounds for the average
estimation error. The principal inertias are a more refined metric of the
correlation between $X$ and $Y$ than, say, mutual information. Furthermore, the
principal inertia components can be used in metrics, such as
$k$-correlation, that share several properties with mutual information (e.g.
convexity).

Furthermore, we also introduced a general method for bounding the average
estimation error of functions of a hidden random variable. This method depends
on the Schur-concavity  of a lower bound for the error-rate function. We proved
that the  $e_\calI(\px,\theta)$ itself is Schur-concave whenever the measure
of information is concave in $\px$. It remains to be shown if
$e_\calI(\px,\theta)$ is Schur-concave for more general measures of information
(such as $k$-correlation), and finding the necessary and sufficient conditions
for Schur-concavity would be of both theoretical and practical interest.

Finally, the creation of bounds for $P_e$ and $P_{e,M}$ given constraints on
different metrics of information is a promising avenue of research. Most
information-theoretic lower bounds for  the average estimation error are based
on mutual information. However, in statistics, a wide range of metrics are used
to estimate  the information between an observed and a hidden
variable. Relating such metrics with the fundamental limits of inference is
relevant for practical applications in both security and machine learning.

\appendices
\section{Proof of Theorem \ref{thm:Bound}}

Theorem  \ref{thm:Bound} follows directly from the next two lemmas.

\begin{lem}
  \label{lem:fanobound1}
Let the marginal distribution $\Px$ and the principal inertias
$\blambda=(\lambda_1,\dots,\lambda_d)$ be
given, where   $d=m-1$. Then
for any $\Pxy\in \mathcal{R}(\Px,\blambda)$, $0\leq \alpha \leq 1$ and $0\leq \beta \leq \px(2)$
\begin{equation*} 
  P_e \geq 1-\beta -
  \sqrt{f_0(\alpha,\Px,\blambda)+\sum_{i=1}^m\left(\left[\px(i)-\beta\right]^+\right)^2},
\end{equation*}
where
\begin{align}
  f_0(\alpha,\Px,\blambda)=&  \sum_{i=2}^{d+1}
  \px(i)(\lambda_{i-1}+c_{i}-c_{i-1})\nonumber\\
  &+\px(1)(c_1 + \alpha) - \alpha\Px^T\Px~,
  \label{eq:f0_def}
\end{align}
and $c_i = \left[ \lambda_{i}-\alpha \right]^+$ for $i=1,\dots,d$ and
$c_{d+1}=0$.
\end{lem}
    
\begin{proof}
Let $X$ and $Y$ have a joint distribution matrix $\Pxy$ with marginal $\px$ and
principal inertias individually bounded by $\blambda =
(\lambda_1,\dots,\lambda_d)$. We assume without loss of generality that
$d=m-1$, where $|\calX|=m$. This can always be achieved by adding inertia components equal to
0.
  
Consider $X\rightarrow Y\rightarrow \hat{X}$, where $\hat{X}$ is the estimate of $X$
from $Y$.  The mapping from $Y$ to $\hat{X}$ can be described without loss of
generality by a
$|\calY|\times|\calX|$ row stochastic matrix, denoted by $\Emat$, where the
$(i,j)$-th entry is the probability $p_{\hat{X}|Y}(j|i)$. The probability of
correct estimation $P_c$ is then
\begin{equation*} 
	P_c = \PR{\hat{X}= X}=\Tr{\Pxxp},
\end{equation*}
where $\Pxxp \defined \Pxy \Emat$.

The matrix $\Pxxp$ can be decomposed according to \eqref{eq:comactDecompPxy},
resulting in 
\begin{align}
  P_c  = \Tr{D_{X}^{1/2}U \Sigmat V^T D_{X'}^{1/2}}= \Tr{\Sigmat  V^T D_{X'}^{1/2} D_{X}^{1/2}U}, \label{eq:PxxpDecomp}
\end{align}
where  
\begin{align*}
  U &= \left[ \Px^{1/2}~~\bu_2~\cdots~\bu_m \right],\\
  V &= \left[ \Pxp^{1/2}~~\bv_2~\cdots~\bv_m  \right],\\
  \Sigmat &= \diag{1,\lambdat_1^{1/2},\dots,\lambdat_d^{1/2}},\\
  D_{X'} & = \diag{\Pxp},
\end{align*}
and $\Ut$ and $\Vt$ are orthogonal matrices. The probability of correct
detection can be written as
\begin{align*}
  P_c &= \Px^T\Pxp + \sum_{k=2}^m\sum_{i=1}^m
  \left(\lambdat_{k-1}\px(i)\pxp(i)\right)^{1/2}u_{k,i}v_{k,i} \nonumber\\
    &= \Px^T\Pxp + \sum_{k=2}^m\sum_{i=1}^m
    \lambdat_{k-1}^{1/2} \tu_{k,i}\tv_{k,i}
\end{align*}
where $u_{k,i}=[\bu_k]_i$, $v_{k,i}=[\bv_k]_i$, $\tu_{k,i}=\px(i)u_{k,i}$ and
$\tv_{k,i}= \pxp(i)v_{k,i}$. Applying the Cauchy-Schwarz inequality twice, we
obtain
\begin{align}
  P_c &\leq \Px^T\Pxp + \sum_{i=1}^m  \left(\sum_{k=2}^m \tv_{k,i}^2
  \right)^{1/2}\left(\sum_{k=2}^m  \lambdat_{k-1} \tu_{k,i}^2
  \right)^{1/2}\nonumber\\
  & =  \Px^T\Pxp +  \sum_{i=1}^m \left(\pxp(i)(1-\pxp(i))\sum_{k=2}^m
  \lambdat_{k-1} \tu_{k,i}^2
  \right)^{1/2}\nonumber\\
  &\leq  \Px^T\Pxp +\left(1-\sum_{i=1}^m \pxp^2(i) \right)^{1/2}\left(
  \sum_{i=1}^m \sum_{k=2}^m  \lambdat_{k-1} \tu_{k,i}^2 \label{eq:LeftBound}
  \right)^{1/2}
\end{align}
Let $\olU=[\bu_2 \cdots \bu_m]$ and $\Sigma =
\diag{\lambdat_1,\dots,\lambdat_d}$. Then
\begin{align}
  \sum_{i=1}^m \sum_{k=2}^m  \lambdat_{k-1} \tu_{k,i}^2 &= \Tr{\Sigma\olU^T
  D_X\olU}\nonumber \\
  &\leq \sum_{k=1}^d\sigma_k \lambdat_k, \nonumber \\
  &\leq \sum_{k=1}^d\sigma_k \lambda_k. \label{eq:prodVonNeum}
\end{align}
where $\sigma_k= \Lambda_k(\olU^T D_X\olU)$. The first inequality follows from the
application of Von-Neumman's trace inequality and the fact that $\olU^T D_X\olU$ is symmetric and positive
semi-definite. The second inequality follows by observing that the principal
inertias satisfy the data processing inequality and, therefore, $\lambdat_k\leq \lambda_k$.

We will now find an upper bound for \eqref{eq:prodVonNeum} by bounding the
eigenvalues $\sigma_k$. First, note that $\olU
~\olU^T =I-\Px^{1/2}\left(\Px^{1/2}\right)^T$ and consequently
\begin{align}
  \sum_{k=1}^d \sigma_k& = \Tr{\olU^T D_X\olU} \nonumber \\
  & = \Tr{D_X\left(I-\Px^{1/2}\left(\Px^{1/2}\right)^T\right)} \nonumber\\
  & = 1-\sum_{i=1}^m\px^2(i)~.\label{eq:sumbound}
\end{align}
Second, note that $\olU^T D_X\olU$ is a principal submatrix of $U^T D_X U$,
formed by removing the first row and columns of $U^T D_X U$. It then follows
from Cauchy's interlacing theorem that
\begin{equation}
  \px(m)\leq \sigma_{m-1} \leq \px(m-1)\leq \dots\leq \px(2)\leq
  \sigma_1 \leq \px(1). \label{eq:cauchyLace}
\end{equation}

Combining the restriction \eqref{eq:sumbound} and \eqref{eq:cauchyLace}, an
upper bound for \eqref{eq:prodVonNeum} can be found by solving the following
linear program
\begin{align}
  \max_{\sigma_i}~~& \sum_{i=1}^d \lambda_i \sigma_i \label{eq:boundLP}\\
  \mbox{subject to}~~&\sum_{i=1}^d \sigma_i =1-\Px^T\Px, \nonumber\\
                        & \px(i+1)\leq \sigma_i \leq
                        \px(i),~i=1,\dots,d~.\nonumber
\end{align}

Let $\delta_i \defined \px(i)-\px(i+1)$ and $\gamma_i \defined \lambda_i\px(i+1)$. The dual of \eqref{eq:boundLP} is
\begin{align}
  \min_{y_i,\alpha}~~& \alpha\left(\px(1)-\Px^T\Px\right)+
  \sum_{i=1}^{m-1}\delta_i y_i+\gamma_i  \label{eq:DualboundLP}\\
  \mbox{subject to}~~&y_i\geq \left[ \lambda_i-\alpha\right]^+,~i=1,\dots,d~.\nonumber
\end{align}
For any given value of $\alpha$, the optimal values of the dual variables $y_i$
in  \eqref{eq:DualboundLP} are
\begin{equation*}
y_i =  \left[ \lambda_i-\alpha\right]^+=c_i,~i=1,\dots,d~.
\end{equation*}
Therefore the linear program \eqref{eq:DualboundLP} is equivalent to
\begin{equation}
  \min_{\alpha}   f_0(\alpha,\Px,\blambda), \label{eq:DualImproved}
\end{equation}
where $f_0(\alpha,\Px,\blambda)$ is defined in the statement of the theorem.

Denote the solution of \eqref{eq:boundLP} by $f_P^*(\Px,\blambda)$ and of
\eqref{eq:DualboundLP} by $f_D^*(\Px,\blambda)$. It follows that
\eqref{eq:prodVonNeum} can be bounded 
\begin{align}
\sum_{k=1}^d\sigma_k \lambda_k &\leq f_P^*(\Px,\blambda) \nonumber \\
                                &=f_D^*(\Px,\blambda) \nonumber\\
                                &\leq f_0(\alpha,\Px,\blambda)~\forall~\alpha
                                \in \Reals. \label{eq:boundSumProd}
\end{align}
We may consider $0\leq\alpha \leq 1$ in
\eqref{eq:boundSumProd} without loss of generality.

Using \eqref{eq:boundSumProd} to bound \eqref{eq:LeftBound}, we find
\begin{equation}
  P_c\leq \Px^T\Pxp + \left[f_0(\alpha,\Px,\blambda)\left(1-\sum_{i=1}^m
  \pxp^2(i) \right)\right]^{1/2} \label{eq:postLPBound}
\end{equation}
The previous bound can be maximized over all possible output distributions
$\pxp$ by solving:
\begin{align}
  \max_{x_i}~~& \left[f_0(\alpha,\Px,\blambda)\left(1-\sum_{i=1}^m
  x_i^2\right)\right]^{1/2} +\sum_{i=1}^m \px(i)x_i  \label{eq:NLP} \\
  \mbox{subject to}~~&\sum_{i=1}^m x_i = 1,\nonumber\\
                        & x_i \geq 0, i=1,\dots,m~.\nonumber
\end{align}
The dual function of \eqref{eq:NLP} over the additive constraint is
\begin{align}
  L(\beta) &=\max_{x_i\geq 0}~~\beta+ \left[f_0(\alpha,\Px,\blambda)\left(1-\sum_{i=1}^m
  x_i^2\right)\right]^{1/2} \nonumber \\
  &\hspace{.5in} +\sum_{i=1}^m (\px(i)-\beta)x_i \nonumber \\
&=\beta + \sqrt{f_0(\alpha,\Px,\blambda) +\sum_{i=1}^m \left(
\left[\px(i)-\beta\right]^+\right)^2}. 
\end{align}
Since $L(\beta)$ is an upper bound of \eqref{eq:NLP} for any $\beta$ and, therefore,
is also an upper bound of  \eqref{eq:postLPBound}, then
\begin{equation}
P_c \leq \beta + \sqrt{f_0(\alpha,\Px,\blambda) +\sum_{i=1}^m \left(
\left[\px(i)-\beta\right]^+\right)^2}. \label{eq:Pc_anyBeta}
\end{equation}
Note that we can consider $0\leq \beta\leq \px(2)$ in \eqref{eq:Pc_anyBeta},
since $L(\beta)$ is increasing for $\beta>\px(2)$. Taking $P_e = 1-P_c$, the result follows.
\end{proof}

The next result tightens the bound introduced in lemma \ref{lem:fanobound1} by
optimizing over all values of $\alpha$.

\begin{lem}
  \label{lem:f0tight}
  Let $f^*_0(\Px,\blambda)\defined\min_{\alpha} f_0(\alpha,\Px,\blambda)$ and
  $k^*$ be defined as in \eqref{eq:kstar}. Then 
  \begin{align}
    f^*_0(\Px,\blambda)=&  \sum_{i=1}^{k^*}\lambda_i \px(i)+
    \sum_{i=k^* + 1}^{m}\lambda_{i-1} \px(i) \nonumber
    \\& - \lambda_{k^*} \Px^T\Px~, 
  \end{align}
  where $\lambda_m=0$.
\end{lem}

\begin{proof}
  Let $\Px$ and $\blambda$ be fixed, and $\lambda_k \leq \alpha \leq
  \lambda_{k-1}$, where we define $\lambda_0=1$ and $\lambda_m =0$.  Then $c_i =\lambda_i-\alpha$
  for $1\leq i\leq k-1$ and $c_i = 0$ for $k\leq i \leq d$ in \eqref{eq:f0_def}.
  Therefore
  \begin{align}
    f_0(\alpha,\Px,\blambda)=&\sum_{i=1}^{k-1} \lambda_i \px(i) +
    \alpha \px(k) \nonumber\\
    & + \sum_{i=k+1}^m \lambda_{i-1}\px(i)-\alpha\Px^T\Px \label{eq:f0_fixaplha}
  \end{align}

Note that \eqref{eq:f0_fixaplha} is convex in $\alpha$, and is decreasing
when $\px(k)-\Px^T\Px\leq 0$ and increasing when   $\px(k)-\Px^T\Px\geq 0$.
Therefore,  $f_0(\alpha,\Px,\blambda)$ is minimized when $\alpha=\lambda_{k}$ such
that $\px(k)\geq \Px^T\Px$ and $\px(k-1)\leq \Px^T\Px$. If $\px(k)-\Px^T\Px\geq
0$ for all $k$ (i.e. $\px$ is uniform), then we can take $\alpha=0$. The result
follows.
\end{proof}

\section{Equivalent characterizations of the principal inertias and
$k$-correlation}
\label{apx:interp}

 In this appendix we discuss two distinct characterizations of the principal
 inertia components. The first characterization  is based on the work of Gebelein
 \cite{gebelein_statistische_1941} and the overview presented in
 \cite{anantharam_maximal_2013}. The second characterization is based on the
 overview presented in \cite{greenacre_geometric_1987}, and is analogous to the
 definition of moments of inertia from  classical mechanics.

\subsection{Correlation characterization}

Let $\calS$ be a collection of random variables defined as
\begin{align*}
  \calS \defined \left\{ (f(X),g(Y)):\right.&\ExpVal{}{f(X)}=\ExpVal{}{g(Y)}=0,\\
  &\left.\ExpVal{}{f^2(X)}=\ExpVal{}{g^2(Y)}=1 \right\}~.
\end{align*}
Then, for  $1\leq k \leq d$, we can compute the principal inertias recursively as
\begin{align*}
  \lambda^{1/2}_{k} &= \max_{(f(X),g(Y))\in\calS_k}\ExpVal{}{f(X)g(Y)},\\
(f_k(X),g_k(Y)) &= \argmax_{(f(X),g(Y))\in\calS_k}\ExpVal{}{f(X)g(Y)},
\end{align*}
where $\calS_1 = \calS$ and
\begin{align*}
    \calS_k = &\left\{ \left(f(X),g(Y) \right)\in \calS :
      \ExpVal{}{f(X)f_i(X)}=0,\right.\\ 
      &\left. \ExpVal{}{g(Y)g_i(Y)}=0,
    i=1,\dots,k-1   \right\}.
\end{align*}
for $2\leq k \leq d$. We can verify that  
$f_k(x)=a_{x,k}/\px(x)$ and $g_k(x)=b_{y,k}/\py(y)$.

\subsection{Spatial characterization}

Let $S$ be defined in equation \eqref{eq:SandP}.
Then the square of the singular values of $S$
are the principal inertias of $\Pxy$ \cite{greenacre_theory_1984}. The
decomposition of $S$ can be interpreted as  the moment of inertia of a set of
masses located in discrete points in space, as described below. We will change
the notation slightly in
this appendix in order to make this analogy clear.

Consider an $n$-dimensional Euclidean space $V$ with a symmetric positive
definite form $Q=D_Y$. For $x,y\in V$ we let  $\brk{x,y}=x^TQy$, $\normQ{x} =
\sqrt{\brk{x,x}}$ and $d(x,y)=\normQ{x-y}$.

Let $x_1,\dots,x_m \in V$, where each point is $x_i = \Ppygx{i}$. We associate to each point $x_i$ a mass
$w_i=\px(i)$, $1\leq i \leq m$. The \textit{barycenter} (center of mass) $\ox$ of the
points $x_1,\dots,x_m$ is simply $\ox=\Py$.  

Let $G=\Pygx$. If we translate the  space so the barycenter $\ox$ is the origin,
the new coordinates of $x_1,\dots,x_m$ are then the rows of $C = G- \ones\ox^T$,
We denote $C^T = [c_1,\dots,c_m]$. We define the \textit{moment of inertia}
$\inertia$ of the collection of $m$ points as the weighted sum of the squared
distances of each point to the barycenter:
\begin{align} 
  \inertia &\defined \sum_{i=1}^{m} w_i d^2(x_i-\ox)
  \label{eq:inertia}\\
  & =  \Tr{D_XCQC^T}. \label{eq:inertia_as_trace}
\end{align}

We now ask: What is the subspace $W^t\in V$ of dimension $t\leq m$ where the
projection of $x_1,\dots,x_m$ has the largest moment of inertia?  To answer this
question we need to determine a basis $a_1,\dots,a_t$ of $W^t$.  This is
equivalent to solving the following optimization:
\begin{align}
 \inertia_t \defined \max_{a_1,\dots,a_t}~ & \sum_{j=1}^d \normEuc{D_X^{1/2}CQa_j}^2
    \label{eq:max_subspaceK}\\
    \mathrm{s.t.}& \normQ{a_j}=1,{a_j}\in V~j=1,\dots,t\nonumber\\
  & \brk{a_i,a_j} =0, ~1\leq i<j\leq t \nonumber
\end{align}
Note that $S = D_X^{1/2}CQ^{1/2}$, and the decomposition in \eqref{eq:SandP} can
be interpreted accordingly. The solution of  \eqref{eq:max_subspaceK} is exactly
the sum of the square of the $t$ largest singular values of $S$, which, in turn,
is equal to $J_t(X;Y)$.

\section*{Acknowledgement}
The authors would like to thank Prof. Shafi Goldwasser and Prof. Yury Polyanskiy
 for the insightful discussions and suggestions throughout the course of
this work.

\bibliography{references}
\bibliographystyle{IEEEtran}

\end{document}